\documentclass[preprint,12pt]{elsarticle}
\usepackage[left=1in,top=1in,right=1in,bottom=1in,letterpaper]{geometry}

\usepackage[usenames]{color}

\usepackage{amsmath,amssymb,amsthm}
\usepackage{latexsym,amsfonts,amscd,amsxtra,amstext}
\usepackage{url}

\newtheorem{theorem}{Theorem}
\newtheorem{definition}{Definition}
\newtheorem{remark}{Remark}
\newtheorem{lemma}{Lemma}

\begin{document}

\begin{frontmatter}

\title{On robust width property for Lasso and Dantzig selector models}

\author[HZ]{Hui Zhang\corref{cor 1}}
\ead{h.zhang1984@163.com}

\cortext[cor 1]{Corresponding author}
\address[HZ]{
College of Science, National University of Defense Technology, Changsha, Hunan, China, 410073}


\begin{abstract}
Recently, Cahill and Mixon completely characterized the sensing operators in many compressed sensing instances with a robust width property. The proposed property allows uniformly stable and robust reconstruction of certain solutions from an underdetermined linear system via convex optimization. However, their theory does not cover the Lasso and Dantzig selector models, both of which are popular alternatives in the statistics community. In this letter, we show that the robust width property can be perfectly applied to these two models as well. Our results solve an open problem left by Cahill and Mixon.
\end{abstract}

\begin{keyword}robust width property; compressed sensing; Lasso model; Dantzig selector model
\end{keyword}
\end{frontmatter}

\section{Introduction}
One of the main assignments of compressed sensing is to understand when it is possible to recover structured solutions to underdetermined systems of linear equations \cite{candes2014math}. During the past decade, there have developed many reconstruction guarantees; well-known concepts include restricted isometry property, null space property, coherence property, dual certificate, and more (the interested readers could refer to \cite{foucart2014math,zhang2012necessary,zhang2014one}). However, none of them is proved necessary for uniformly stable and robust reconstruction. Recently, Cahill and Mixon in \cite{jameson2014robust} introduced a new notion--robust width property, which completely characterizes the sensing operators in many compressed sensing instances. They restricted their attention into the following constrained optimization problem:
\begin{equation}\label{BP}
\min \|x\|_\sharp,~~\textrm{subject to} ~~ \|\Phi x-y\|_2\leq \epsilon  \tag{$Q_{\epsilon}$}
\end{equation}
such that their theory can not cover the Lasso and Dantzig selector models, both of which are popular alternatives in the statistics community. Here, $\|\cdot\|_\sharp$ is some norm used to promote certain structured solutions, operator $\Phi$ and data $y$ are given, and $\epsilon$ measures the error. In this letter, we extend their results to two other probably more popular optimization problems of the Lasso/Basis Pursuit and Dantzig selector types. Our derived results completely solve an open problem left by Cahill and Mixon and hence prove that the notion of robust width is indeed a ubiquitous property. In the following, we recall some notations appeared in the paper \cite{jameson2014robust}.

Let $x^\natural$ be some unknown member of a finite-dimensional Hilbert space $\mathcal{H}$, and let $\Phi :\mathcal{H}\rightarrow \mathbb{F}^M$ denote some known linear operator, where $\mathbb{F}$ is either $\mathbb{R}$ or $\mathbb{C}$. Subset $\mathcal{A}\subseteq \mathcal{H}$ is a particular subset that consists of some type of structured members.  $B_\sharp$ is the unit $\sharp$-ball.

\section{Robust width}
The robust width property was formally proposed in \cite{jameson2014robust}. We write down the definition and its equivalent form as follows.
\begin{definition}(\cite{jameson2014robust})
We say a linear operator $\Phi :\mathcal{H}\rightarrow \mathbb{F}^M$ satisfies the $(\rho,\alpha)$-robust width property over $B_\sharp$ if
$$ \|x\|_2\leq \rho \|x\|_\sharp$$
for every $x\in \mathcal{H}$ such that $\|\Phi x\|_2<\alpha \|x\|_2$; or equivalently if
$$\|\Phi x\|_2\geq \alpha \|x\|_2 $$
for every $x\in \mathcal{H}$ such that $\|x\|_2>\rho \|x\|_\sharp$.
\end{definition}
Here, we would like to point out the definition above is not completely new. In fact, when restricted to the case of $\ell_1$-minimization, it reduces to the $\ell_1$-constrained minimal singular value property which was originally defined in \cite{tang2011performance}.

\begin{definition}
For any $k \in \{1,2,\cdots, N\}$ and matrix $\Phi\in \mathbb{R}^{M\times N}$, define
the $\ell_1$-constrained minimal singular value of $\Phi$ by
\begin{equation*}r_k(\Phi)=\min_{x\neq 0, x\in S_k}\frac{\|\Phi x\|_2}{\|x\|_2}\end{equation*}
where $S_k=\{x\in \mathbb{R}^N: \|x\|_1\leq \sqrt{k}\|x\|_2\}$. If $r_k(\Phi)>0$, then we say $\Phi$ satisfies the $\ell_1$-constrained minimal singular value property with $r_k(\Phi)$.
\end{definition}

Work \cite{zhang2012constrained} exploited the geometrical aspect of the $\ell_1$-constrained minimal singular value property.

\section{Main results}
 We first introduce the definition of compressed sensing space.
\begin{definition}(\cite{jameson2014robust})\label{cs}
A compressed sensing space $(\mathcal{H},\mathcal{A}, \|\cdot\|_\sharp )$  with bound $L$ consists of a finite-dimensional Hilbert space $\mathcal{H}$, a subset $\mathcal{A}\subseteq \mathcal{H}$, and a norm $\|\cdot\|_\sharp$ on $\mathcal{H}$ with following properties:

 (i) $0\in \mathcal{A}$.

 (ii) For every $a\in \mathcal{A}$ and $v\in \mathcal{H}$, there exists a decomposition $v= z_1+z_2$ such that
 $$\|a+z_1\|_\sharp =\|a\|_\sharp+\|z_1\|_\sharp, ~~~~\|z_2\|_\sharp\leq L\|v\|_2.$$

\end{definition}

The subdifferential $\partial f(x)$ of a convex function $f$ at $x$ is the set-valued operator \cite{bauschke2011convex} given by
$$\partial f(x)=\{u\in\mathcal{H}: f(y)\geq f(x)+\langle u, y-x\rangle, \forall  y\in \mathcal{H}\}.$$
The following lemma will be useful to establish our main results.
\begin{lemma}\label{lem1}
Let $\|\cdot\|_\diamond$ be the dual norm of $\|\cdot\|_\sharp$ on $\mathcal{H}$. If $u\in \partial \|x\|_\sharp$, then $\|u\|_\diamond\leq 1$. If $x\neq 0$ and $u\in \partial \|x\|_\sharp$, then $\|u\|_\diamond= 1$.
\end{lemma}
\begin{proof}
From the convexity of $\|\cdot\|_\sharp$ and the subdifferential definition, for any $u\in \partial \|x\|_\sharp$ and $v\in \mathcal{H}$ it holds
$$ \|v\|_\sharp \geq \|x\|_\sharp +\langle u, v-x\rangle.$$
Set $v=0$ and $v=2x$ to get $\langle u, x\rangle \geq \|x\|_\sharp $ and $\langle u, x\rangle \leq \|x\|_\sharp $ respectively. This implies $\langle u, x\rangle = \|x\|_\sharp $ and hence $\langle u, v\rangle \leq \|v\|_\sharp $. Similarly, by taking $-v\in \mathcal{H}$, we can get $-\langle u, v\rangle \leq \|v\|_\sharp $. Thus,  $|\langle u, v\rangle| \leq \|v\|_\sharp $. Therefore, $$\|u\|_\diamond=\sup_{\|v\|_\sharp\leq 1}|\langle u, v\rangle|\leq \sup_{\|v\|_\sharp\leq 1 }\|v\|_\sharp\leq 1.$$
When $x\neq 0$, by the Cauchy-Schwartz inequality we get that $\|x\|_\sharp=\langle u, x\rangle\leq \|x\|_\sharp \|u\|_\diamond$ and hence $\|u\|_\diamond\geq 1$. So it must have $\|u\|_\diamond= 1$.
\end{proof}

Now, we state the characterization of uniformly stable and robust reconstruction via the Lasso/Basis Pursuit type model by utilizing the $(\rho,\alpha)$-robust width property.
\begin{theorem}
For any CS space $(\mathcal{H},\mathcal{A}, \|\cdot\|_\sharp )$  with bound $L$ and any linear operator $\Phi :\mathcal{H}\rightarrow \mathbb{F}^M$, the following are equivalent up to constants:

(a) $\Phi$ satisfies the $(\rho,\alpha)$-robust width property over $B_\sharp$.

(b) For every $x^\natural \in\mathcal{H}, \kappa\in (0,1), \lambda>0$ and $\omega\in\mathbb{F}^M$ satisfying $\|\Phi^T\omega\|_\diamond\leq \kappa\lambda$, any solution $x^*$ to the unconstrained optimization model
\begin{equation}\label{Lasso}
\min \frac{1}{2}\|\Phi x-(\Phi x^\natural +\omega)\|_2^2+\lambda \|x\|_\sharp \tag{$P_{\lambda}$}
\end{equation}
satisfies $\|x^*-x^\natural \|_2 \leq C_0 \|x^\natural -a \|_\sharp +C_1 \cdot \lambda$ for every $a\in \mathcal{A}$.

In particular, (a) implies (b) with
$$C_0=\left( \frac{1-\kappa}{2\rho}-L\right)^{-1},~~~~ C_1=\frac{1+\kappa}{\alpha^2\rho}$$
provided $\rho<\frac{1-\kappa}{2L}$. Also, (b) implies (a) with
$$ \rho =2C_0, ~~~~\alpha=\frac{\kappa}{2\tau C_1},$$
where $\tau=\sup_{\|x\|_\sharp\leq 1}\|\Phi x\|_2$.
\end{theorem}

\begin{proof}
Let $z=x^*-x^\natural$. We divide the proof of $(a)\Rightarrow (b)$ into four steps. They are partially inspired by \cite{candes2011tight} and \cite{jameson2014robust}.

\textbf{Step 1:} Prove the first relationship:
 \begin{equation}\label{step1eq}
\|x^*\|_\sharp-\kappa \|z\|_\sharp \leq \|x^\natural\|_\sharp.
\end{equation}
Since $x^*$ is a minimizer to \eqref{Lasso}, we have
$$\frac{1}{2}\|\Phi x^*-(\Phi x^\natural +w)\|_2^2+\lambda \|x^*\|_\sharp\leq \frac{1}{2}\|\Phi x^\natural -(\Phi x^\natural +w)\|_2^2+\lambda \|x^\natural \|_\sharp.$$
Hence,
$$\frac{1}{2}\|(\Phi x^*-\Phi x^\natural)-w\|_2^2+\lambda \|x^*\|_\sharp\leq \frac{1}{2}\|w\|_2^2+\lambda \|x^\natural \|_\sharp.$$
Rearrange terms to give
$$\lambda \|x^*\|_\sharp\leq -\frac{1}{2}\|\Phi (x^* -x^\natural)\|_2^2+ \langle\Phi(x^*-x^\natural),  w\rangle +\lambda \|x^\natural \|_\sharp \leq \langle x^*-x^\natural, \Phi^T w\rangle +\lambda \|x^\natural \|_\sharp.$$ By the Cauchy-Schwartz inequality and the condition $\|\Phi^Tw\|_\diamond\leq \kappa\lambda$,  we obtain that
$$\langle x^*-x^\natural, \Phi^T w\rangle\leq \|x^*-x^\natural\|_\sharp \|\Phi^Tw\|_\diamond\leq \kappa\lambda\|x^*-x^\natural\|_\sharp.$$
Thus, $\lambda \|x^*\|_\sharp\leq \kappa\lambda\|x^*-x^\natural\|_\sharp+ \lambda \|x^\natural \|_\sharp$ from which the first relationship follows.

\textbf{Step 2:} Prove the second relationship:
 \begin{equation}\label{step2eq}
\|z\|_\sharp\leq \frac{2}{1-\kappa}\|x^\natural -a\|_\sharp +\frac{2L}{1-\kappa}\|z\|_2.
\end{equation}
Pick $a\in \mathcal{A}$, and decompose $z=x^*-x^\natural=z_1+z_2$ according to the property (ii) in Definition \ref{cs} so that $\|a+z_1\|_\sharp =\|a\|_\sharp+\|z_1\|_\sharp$ and $\|z_2\|_\sharp\leq L\|z\|_2.$ In light of \eqref{step1eq}, we derive that
\begin{align*}
 \|a\|_\sharp +\|x^\natural -a\|_\sharp &\geq \|x^\natural\|_\sharp  \\
  &\geq  \|x^*\|_\sharp-\kappa \|z\|_\sharp \\
  & =\|x^\natural+(x^*-x^\natural)\|_\sharp-\kappa \|x^*-x^\natural\|_\sharp\\
  & =\|a + (x^\natural-a)+z_1+z_2\|_\sharp-\kappa \|z_1+z_2\|_\sharp\\
  &\geq \|a+z_1\|_\sharp-\| x^\natural-a \|_\sharp-(1+\kappa)\|z_1\|_\sharp-\kappa \|z_2\|_\sharp\\
  &=\|a\|_\sharp + \|z_1\|_\sharp-\| x^\natural-a \|_\sharp-(1+\kappa)\|z_2\|_\sharp-\kappa \|z_1\|_\sharp\\
  &=\|a\|_\sharp + (1-\kappa)\|z_1\|_\sharp-\| x^\natural-a \|_\sharp-(1+\kappa)\|z_2\|_\sharp.
\end{align*}
Rearrange terms to give $$\|z_1\|_\sharp\leq \frac{2}{1-\kappa}\|x^\natural -a\|_\sharp+\frac{1+\kappa}{1-\kappa}\|z_2\|_\sharp$$ which implies
$$\|z\|_\sharp\leq \|z_1\|_\sharp+ \|z_2\|_\sharp\leq \frac{2}{1-\kappa}\|x^\natural -a\|_\sharp+\frac{2}{1-\kappa}\|z_2\|_\sharp.$$
Thus, the second relationship follows by invoking $\|z_2\|_\sharp\leq L\|z\|_2.$

\textbf{Step 3:} Derive the upper bound:
 \begin{equation}
\|\Phi z\|_2^2\leq (1+\kappa)\lambda \|z\|_\sharp.
\end{equation}
The optimality condition of \eqref{Lasso} reads
$$\Phi^T(\Phi x^\natural +w-\Phi x^*)\in \lambda\cdot\partial \|x^*\|_\sharp.$$
By using Lemma \ref{lem1}, we get $\|\Phi^T(\Phi x^\natural +w-\Phi x^*)\|_\diamond \leq \lambda$. Thus,
\begin{subequations}
\begin{align*}
\|\Phi^T\Phi z\|_\diamond  &= \|\Phi^T\Phi (x^*-x^\natural)\|_\diamond  \\
  &\leq  \|\Phi^T(\Phi x^*-\Phi x^\natural-w)\|_\diamond  +\|\Phi^Tw\|_\diamond  \\
  & \leq \lambda+\kappa \lambda=(1+\kappa)\lambda.
\end{align*}
\end{subequations}
Therefore,
$$\|\Phi z\|_2^2  = \langle z, \Phi^T\Phi z\rangle
  \leq  \|z\|_\sharp\cdot \|\Phi^T\Phi z\|_\diamond\leq  (1+\kappa)\lambda \|z\|_\sharp,$$
  where the first inequality follows from the Cauchy-Schwartz inequality.

\textbf{Step 4:} Finish the proof.  Assume
$ \|z\|_2> C_0\cdot \|x^\natural -a \|_\sharp $, since otherwise we are done. In light of \eqref{step2eq}, we obtain
$$\|z\|_\sharp<\left[\frac{2}{C_0(1-\kappa)}+\frac{2L}{1-\kappa}\right]\|z\|_2=\rho^{-1}\|z\|_2,$$
i.e., $\|z\|_2>\rho\|z\|_\sharp$. By the $(\rho,\alpha)$-robust width property of $\Phi$, we have
$\|\Phi z\|_2\geq \alpha \|z\|_2$. Utilizing the upper bound of $\|\Phi z\|_2^2$ in Step 3, we derive that
$$\alpha^2\|z\|_2^2\leq \|\Phi z\|_2^2\leq (1+\kappa)\lambda \|z\|_\sharp<\frac{(1+\kappa)\lambda}{\rho}\|z\|_2.$$
Thus, $$\|z\|_2\leq \frac{(1+\kappa)\lambda}{\alpha^2\rho}=C_1\cdot \lambda\leq C_0 \|x^\natural -a \|_\sharp +C_1\cdot \lambda.$$
This completes the proof of $(a)\Rightarrow (b)$.

The proof of $(b)\Rightarrow (a)$. Pick $x^\natural$ such that $\|\Phi x^\natural\|_2<\alpha\|x^\natural\|_2$. By the expression of $\tau=\sup_{\|x\|_\sharp\leq 1}\|\Phi x\|_2$ and using the Cauchy-Schwartz inequality, we derive that
\begin{subequations}
\begin{align*}
\tau\cdot \alpha\|x^\natural\|_2 & >\tau\cdot \|\Phi x^\natural\|_2 =  \sup_{\|x\|_\sharp\leq 1}\|\Phi x\|_2 \cdot  \|\Phi x^\natural\|_2 \\
  & \geq \sup_{\|x\|_\sharp\leq 1}\langle \Phi x, \Phi x^\natural\rangle = \sup_{\|x\|_\sharp\leq 1}\langle x, \Phi ^T \Phi x^\natural\rangle\\
  &= \|\Phi ^T \Phi x^\natural\|_\diamond.
\end{align*}
\end{subequations}
Let $\lambda=\kappa^{-1}\tau\alpha \|x^\natural\|_2$ and $\omega=-\Phi x^\natural$. Then, we have
$$\kappa\lambda = \tau\cdot \alpha \|x^\natural\|_2\geq \|\Phi ^T \Phi x^\natural\|_\diamond=\|\Phi^T w\|_\diamond,$$
which implies that the choosing of $\lambda$ and $\omega$ satisfies the constrained condition $\|\Phi^T w\|_\diamond\leq \kappa\lambda$.  Thereby, we can take $\omega=-\Phi x^\natural$ and hence conclude that $x^*=0$ is a minimizer of \eqref{Lasso}. Thus,
$$\|x^\natural\|_2=\|x^*-x^\natural\|_2\leq C_0\|x^\natural\|_\sharp+C_1\lambda = C_0\|x^\natural\|_\sharp+C_1\kappa^{-1}\tau\alpha \|x^\natural\|_2.$$
Take $\alpha =\frac{\kappa}{2\tau C_1}$ and $\rho=2C_0$ and rearrange terms to give $$\|x^\natural\|_2\leq \frac{C_0}{1-C_1\kappa^{-1}\tau\alpha}\|x^\natural\|_\sharp=\rho \|x^\natural\|_\sharp.$$
So the $(\rho,\alpha)$-robust width property of $\Phi$ holds.
\end{proof}

\begin{remark}
In the paper \cite{jameson2014robust}, to obtain a corresponding result for \eqref{BP}, it suffices for $\|\cdot\|_\sharp$ to satisfy:

(i)$\|x\|_\sharp \geq \|0\|_\sharp$ for every $x\in \mathcal{H}$, and

(ii) $\|x+y\|_\sharp \leq \|x\|_\sharp +\|y\|_\sharp$ for every $x, y\in\mathcal{H}$.

In contrast, Theorem 1 not only requires (i) and (ii) above, but also utilizes the convexity of $\|\cdot\|_\sharp$ and its dual norm. The additional requirement of convexity excludes the cases of nonconvex $\|\cdot\|_\sharp$. For example, the case of
$$\|x\|_\sharp=\|x\|_p^p:=\sum_{i=1}^N|x_i|^p, ~~0<p<1$$
is not covered by Theorem 1.

\end{remark}

With very similar arguments, we can show the following theorem which characterizes the uniformly stable and robust reconstruction via the Dantzig type model by utilizing the $(\rho,\alpha)$-robust width property.
\begin{theorem}
For any CS space $(\mathcal{H},\mathcal{A}, \|\cdot\|_\sharp )$  with bound $L$ and any linear operator $\Phi :\mathcal{H}\rightarrow \mathbb{F}^M$, the following are equivalent up to constants:

(a) $\Phi$ satisfies the $(\rho,\alpha)$-robust width property over $B_\sharp$.

(b) For every $x^\natural \in\mathcal{H}, \lambda>0$ and $\omega \in\mathbb{F}^M$ satisfying $\|\Phi^T\omega\|_\diamond\leq \lambda$, any solution $x^*$ to the following optimization model
\begin{equation}\label{Dant}
\min \|x\|_\sharp, ~~~\textrm{subject to}~~  \|\Phi^T(\Phi x-(\Phi x^\natural +\omega))\|_\diamond\leq \lambda \tag{$R_{\lambda}$}
\end{equation}
satisfies $\|x^*-x^\natural \|_2 \leq C_0 \|x^\natural -a \|_\sharp +C_1 \cdot \lambda$ for every $a\in \mathcal{A}$.

In particular, (a) implies (b) with
$$C_0=\left( \frac{1}{2\rho}-L\right)^{-1},~~~~ C_1=\frac{2}{\alpha^2\rho}$$
provided $\rho<\frac{1}{2L}$. Also, (b) implies (a) with
$$ \rho =2C_0, ~~~~\alpha=\frac{\kappa}{2\tau C_1},$$
where $\tau=\sup_{\|x\|_\sharp\leq 1}\|\Phi x\|_2$.
\end{theorem}

\begin{proof} The proof below follows from the pattern used for that of Theorem 1. Let $z=x^*-x^\natural$.

\textbf{Step 1: }Since $x^*$ is a minimizer of \eqref{Dant}, it holds that $\|x^*\|_\sharp\leq \|x^\natural\|_\sharp$. Now, repeat the argument for Step 2 in the proof of Theorem 1 to give
 \begin{equation*}
\|z\|_\sharp\leq 2\|x^\natural -a\|_\sharp +2L\|z\|_2.
\end{equation*}

\textbf{Step 2:} Prove the upper bound:
 \begin{equation*}
\|\Phi z\|_2^2\leq 2\lambda \|z\|_\sharp.
\end{equation*}
This follows from that
$$\|\Phi^T\Phi z\|_\diamond \leq \|\Phi^T(\Phi x^*-(\Phi x^\natural +w))\|_\diamond + \|\Phi^Tw\|_\diamond\leq 2 \lambda$$
and
$$\|\Phi z\|_2^2  = \langle z, \Phi^T\Phi z\rangle\leq  \|z\|_\sharp\cdot \|\Phi^T\Phi z\|_\diamond.$$
The remained proof of $(a)\Rightarrow (b)$ follows by repeating the argument for Step 4 in the proof of Theorem 1.

The proof of $(b)\Rightarrow (a)$. Pick $x^\natural$ such that $\|\Phi x^\natural\|_2<\alpha\|x^\natural\|_2$. Let $\lambda=\tau\alpha \|x^\natural\|_2$ and $\omega=-\Phi x^\natural$. We have proved in the proof of Theorem 1 that such choosing of $\lambda$ and $\omega$ satisfies the constrained condition of $\|\Phi^Tw\|_\diamond\leq \lambda$ and hence $x^*=0$ is the unique minimizer of \eqref{Dant}. The remained proof of $(b)\Rightarrow (a)$ follows by repeating the corresponding part in the proof of Theorem 1.
\end{proof}
Note that the convexity of $\|\cdot\|_\sharp$ is not involved in the proof of Theorem 2.

\section*{Acknowledgements}
The author would like to thank Dr. Jameson Cahill for his communication and anonymous reviewers for their valuable comments, with which great improvements have been made in this manuscript. The work is supported by the National Science
Foundation of China (No.11501569 and No.61571008).


\end{document}